\documentclass[11pt]{article}

\usepackage{fullpage}
\newcommand {\ignore} [1] {}

\usepackage{geometry}
\usepackage{graphicx}

\usepackage{graphics}
\usepackage{color}
\usepackage{amsfonts,amsxtra}
\usepackage{amsmath,amssymb}
\usepackage{float}
\usepackage[english]{babel}
\usepackage{etex}
\usepackage[linesnumbered,boxed,vlined,ruled]{algorithm2e}
\usepackage{times}

\usepackage{authblk}

\renewcommand{\Pr}[1]{\mbox{\rm\bf Pr}\left[#1\right]}
\newcommand{\Ex}[1]{\mbox{\rm\bf E}\left[#1\right]}

\floatstyle{ruled}
\newfloat{Figure}{t}{fig}
\newcommand{\cE}{{\cal E}}
\newcommand{\cB}{{\cal B}}
\newcommand{\cS}{{\cal S}}
\newcommand{\cF}{{\cal F}}

\newenvironment{proof}{\noindent   {\bf Proof.}}{\hspace*{\fill}$\Box$\par\vspace{2mm}}

\newtheorem{lemma}{Lemma}
\newtheorem{theorem}{Theorem}
\newtheorem{corollary}{Corollary}
\newtheorem{claim}{Claim}

\newcommand{\localalgo}{LocalHypColoring}

\title{Deterministic coloring algorithms in the LOCAL model\thanks{This paper has been submitted to ACM-SIAM SODA 2020. Date of this version: 10 July 2019.}}

\author[1]{Dariusz R. Kowalski}
\author[2]{Piotr Krysta}

\affil[1]{Augusta University, Augusta, GA, USA}
\affil[ ]{{\tt darek.liv@gmail.com}}
\affil[2]{University of Liverpool, Liverpool, UK}
\affil[ ]{{\tt pkrysta@liverpool.ac.uk}}

\date{}

\begin{document}

\maketitle
\setcounter{page}{0}
\thispagestyle{empty}

\begin{abstract}
We study the problem of bi-chromatic coloring of hypergraphs in the LOCAL distributed model of
computation. This problem can easily be solved by a randomized local algorithm with no communication. However, it is not known how to solve it deterministically with only a polylogarithmic number of communication rounds. In this paper we indeed design such a deterministic algorithm that solves this problem with polylogarithmic number of communication rounds. This is an almost exponential improvement on the previously known deterministic local algorithms for this problem. Because the bi-chromatic coloring of hypergraphs problem is known to be complete in the class of all locally checkable graph problems, our result implies deterministic local algorithms with polylogarithmic number of communication rounds for all such problems for which an efficient randomized algorithm exists. This solves one of the fundamental open problems in the area of local distributed graph algorithms. By reductions due to Ghaffari, Kuhn and Maus [STOC 2017] this implies such polylogarithmically efficient deterministic local algorithms for many graph problems.
\end{abstract}


\newpage

\section{Introduction}

Suppose that we are given a bipartite graph $G = (U \cup V, E)$ such that $|U| = |V| = n$ and the degree of each $w \in U$ is $\delta/2 \leq \deg(w) \leq \delta$, where $\delta=\delta(n)$ satisfies $c \ln^2 (n) \leq \delta \leq \log^{O(1)}(n)$ for some sufficiently large constant $c >0$.
   The goal is to color the vertices of $V$ with two colors red and blue such that for each $w \in U$ its neighborhood $N(w) \subseteq V$ has at least one vertex colored blue and at least one vertex colored red. A neighborhood $N(w)$ is called {\em bi-chromatic} if it has this property. The coloring is called bi-chromatic if $N(w)$ is bi-chromatic for every $w \in U$, and it is called monochromatic otherwise. Note that here the neighborhood $N(w)$ of $w$ does not contain $w$ itself.

The trivial randomized algorithm is to color each $v \in V$ independently and uniformly at random red with probability $1/2$ and blue with probability $1/2$. This will ensure that each neighborhood $N(w)$ for each $w \in U$ is bi-chromatic with high probability. The goal is to achieve the bi-chromatic coloring by a deterministic algorithm that only uses $poly \log(n)$ number of rounds (in the distributed model LOCAL). For instance, the above randomized algorithm does not need any communication rounds.

Although the coloring problem on graphs is more directly relevant to the LOCAL model of communication and computation (clearly defined computation entities and bidirectional links), it could be equivalently formulated as the hypergraph coloring problem as follows. Let $H = (V, \cE)$ be a hypergraph whose set of vertices is $V$ and set of hyperedges is $\cE$ such that $\cE = \{N(w): w \in U\}$. Our problem is to color the vertices $V$ of hypergraph $H$ with two colors, red and blue, such that each hyperedge $N(w) \in \cE$ of $H$ is bi-chromatic. 
The entities in the LOCAL model are both hyperedges and vertices, and (bidirectional) communication links are between a hyperedge and its vertices. In the remainder we will follow this abstraction of the problem.

\paragraph{Our results and technical contributions.} One of the fundamental open problems in the area of LOCAL distributed graph problems is the question of whether randomization is required for efficient symmetry breaking. Many basic graph computational problems admit efficient randomized local algorithms which only use polylogarithmic communication. On the other hand the best known deterministic local algorithms are almost exponentially worse with respect to communication compared to randomized algorithms. Examples of such problems include hypergraph bi-chromatic coloring (a.k.a.~hypergraph $2$-coloring), conflict-free hypergraph multicoloring, low-diameter ordering, and many other, see, e.g., \cite{GKM17}.

Ghaffari, Kuhn and Maus \cite{GKM17} showed that the bi-chromatic hypergraph coloring problem is complete with respect to the above fundamental open problem in the sense that if there is an efficient deterministic distributed algorithm with polylogarithmic communication for this problem, then there is such an algorithm for all locally checkable graph problems for which an efficient randomized algorithm exists. In this paper, we resolve this fundamental open problem and design a deterministic local algorithm for the bi-chromatic hypergraph coloring problem with $O(\log^2 n)$ communication rounds. By the results of Ghaffari, Kuhn, and Maus \cite{GKM17}, this implies such efficient deterministic local algorithms for many problems, including the local splitting problem, conflict-free hypergraph multicoloring, low-diameter network decomposition into clusters with small chromatic number, low-diameter ordering, approximating general covering or packing integer linear programs. In the complexity-theoretic language of the paper \cite{GKM17} our result shows that P-LOCAL $=$ P-SLOCAL. We would like to emphasize that the best known deterministic local algorithms for all those problems use roughly $2^{O(\sqrt{\log n})}$ communication rounds, where $n$ is typically the number of nodes in the input (hyper)graph, and we improve this to $O(\log^{O(1)} n)$ communication rounds, which is an almost exponential improvement.

We build our new solution on the previous approaches to hypergraph $2$-coloring: categorizing hyperedges with respect to being biased towards specific color, special consideration of biased components, introduced by Beck, and derandomization techniques in the context of the PRAM model proposed by Alon. When using random colorings as a base, the probability of obtaining monochromatic coloring in Beck's and Alon's approaches was enough to apply derandomization in the PRAM model, but not sufficient to do it in the LOCAL model. The reason is that in PRAM processes may access quickly any stored information (though it is sometimes non-trivial to find out by a process where it is), and thus could verify whether there is any large component in the hypergraph which is incorrectly colored. Such a verification is not physically possible in the LOCAL model. We overcome this obstacle in three steps. First, we substantially lower the probability of {\em partial} failure when starting from random colorings. As this was not enough, we introduced consecutive phases, based on independent random colorings, and a way to color and re-color the hypergraph consistently phase by phase. Finally, we analyzed and derandomized this process with respect not to all potential hypergraph topologies (as there are too many of them), but their representatives being $2,3$-skeletons. This way we showed that there exists a deterministic algorithm, such that for any parameter $n$ it can locally compute $x=\Theta(\log^2 n)$ deterministic colorings and apply them in $x$ phases to gradually color any hypergraph of $n$ hyperedges and $n$ vertices in a bi-chromatic way.

\paragraph{Previous and related work.} Our work is closely related to two main strands of research. One is, on (derandomized) algorithms for the Lov\'{a}sz Local Lemma, whose representative problem is the hypergraph bichromatic coloring. Already Erd\H{o}s and Lov\'{a}sz showed in \cite{EL75} a condition for the existence of a bi-chromatic coloring in a hypergraph by proving a version of what is now known as the Lov\'{a}sz Local Lemma (LLL). Beck \cite{Beck91} provided the first algorithmic version of LLL, by giving a randomized sequential algorithm which he showed how to derandomize. Beck's flagship application was also the hypergraph bi-chromatic coloring. Alon \cite{Alon91}, buidling on Beck's approach, introduced a new randomized algorithmic version of LLL on  hypergraph bi-chromatic coloring and showed how to derandomize it that led to a deterministic parallel algorithm on the PRAM model. These randomized, deterministic sequential and parallel algorithms for LLL were improved in \cite{CS,Heupler2013}. Recent work on distributed algorithms for LLL include \cite{CPS14,GHK18}. The other research strand is on efficient (de)randomized algorithms for graph problems in the local model. A recent paper \cite{GKM17} nicely unifies this research area by presenting a general complexity theoretic study; see further references there for an extensive list of previous results. A recent paper by Ghaffari and Kuhn \cite{GK19} gives improved randomized local algorithms for a broad class of graph problems (studied in \cite{GKM17}) with improved success probability. This paper also presents derandomization results for some of these algorithms. Finally, very recently, Bamberger et al.~\cite{BamGKMU19} design a deterministic local algorithm for the hypergraph bichromatic coloring problem with small number of communication rounds if the frequency of vertices in the hyperedges is bounded.

\section{Model and preliminaries}

\subsection{Local communication and computation model}

In the LOCAL model of distributed computing,
there are $N$ computational units connected by links.
There are two most popular communication topologies:
all-to-all (clique) and the topology corresponding
to the input data. We follow the second setting,
which is more challenging due to the limited knowledge
not only about the computational input, but also about 
the communication structure.
In our case - the communication topology
is a bipartite graph of $G=(U,V,E)$ with $N=2n$
computational units, such that $n$ units in $V=\{1,\ldots,n\}$ represent
vertices of the input while $n$ units in $U=\{n+1,\ldots,2n\}$ represent
hyperedges. 
Two units are connected by an edge -- communication link -- if the the vertex represented by one unit belongs to the hyperedge represented by another one.
In the remainder, we will be calling the units representing
vertices -- vertices, and those representing hyperedges -- hyperedges.

Communication happens in synchronous (communication) rounds.
In a single round, each vertex and hyperedge can send/receive a message to/from their neighbors in $G$ and perform
a local computation.
In particular, two intersecting hyperedges can exchange 
information in two rounds, by sending it to a common vertex which forwards it to the other peer in the second round.
Messages sent in a round are received by the neighbors of the senders in the same round.

We consider the problem of finding bi-chromatic coloring
of all hyperedges in a given hypergraph.
As an input, each unit, vertex or hyperedge, gets parameter $n$,
which is both the number of vertices and hyperedges,
and the list of hyperedges it belongs to in case of a vertex or the list of vertices that belong to it in case of a hyperedge.
For a given input parameter $n$, we call all hypergraphs of $n$
vertices and $n$ hyperedges, each containing $\delta$ vertices,
{\em admissible}.

The primary complexity measure in the LOCAL model is a 
{\em round complexity}, defined as the number of communication rounds to accomplish the task.

\noindent
{\bf Assumption:} 
We assume without loss of generality that each hyperedge $f=N(w)$, for any $w \in U$, has the same size $\delta$. If this is not true we simply replace each hyperedge $N(w)$ by {\em an arbitrary} subset of its vertices of size precisely $\delta$. 
Using the same argument, we could also assume further that $\delta = c\cdot \log n$, for a sufficiently
large constant $c>0$.

\subsection{Notation and previous approaches to bi-chromatic coloring of hypergraphs}

In this section we will introduce a useful notation and review some of the previously known tools for bi-chromatic coloring of hypergraphs that date back to Beck \cite{Beck91} and Alon \cite{Alon91}. We will 
build on them to define our new approach and its analysis.

We will describe now some steps of Alon's extension of Beck's algorithm for computing a bi-chromatic coloring that are relevant to our approach. The first phase of the algorithm colors each vertex $v \in V$ independently and uniformly at random with color red or blue with equal probability $1/2$. If after that there are monochromatic hyperedges, the second phase of the algorithm re-colors some of the vertices of $V$ such that finally $H$ becomes bi-chromatic. The following assumption ensures, by the Lov\'{a}sz Local Lemma, that there exists a bi-chromatic coloring in $H$:
$$
   2 e (d+1) \leq 2^{\alpha \delta} 
   \ ,
$$ 
where 
no hyperedge in $H$ intersects more than $d$ other hyperedges, and $0 < \alpha < 1/2$ is some constant. Observe that in our case $d \leq n$. Let $h(\cdot)$ denote the binary entropy function: $h(x) = -x \log_2(x) - (1-x) \log_2(1-x)$.


Suppose that we are given a random coloring of the first phase. We call a hyperedge $N(w) \in \cal{E}$ {\em biased} if at most $\alpha \delta$ of its vertices in $N(w)$ are red or at most $\alpha \delta$ of its vertices in $N(w)$ are blue. Let $\cB$ be the set of all biased hyperedges. It is easy to see that then the probability that a given fixed hyperedge is biased is at most 
$$
2 \sum_{i \leq \alpha \delta} {\delta \choose i}/2^{\delta} \leq 2 \cdot 2^{(h(\alpha)-1)\delta}
\ .
$$ 
Let us define $p = 2 \cdot 2^{(h(\alpha)-1)\delta}$. Let now $G(H)$ be the intersection graph of hypergraph $H$, that is, the vertices of $G(H)$ are the hyperedges of $H$ and any two vertices of $G(H)$ are adjacent if and only if their corresponding hyperedges intersect in $H$; in order to avoid confusion, we will call vertices of $G(H)$ hyperedges, which they indeed correspond to. We call a set $T$ of hyperedges in $G(H)$ a {\em connected $1,2$-component} if $T$ is the set of hyperedges in a connected subgraph in the square of $G(H)$; that is, $T$ is {\em connected $1,2$-component} if for any two hyperedges $x, y \in T$, there is a path between them in the intersection graph $G(H)$ such that among any two consecutive hyperedges on the path at least one belongs to $T$. We also refer as {\em $1,2$-component} to a collection of disjoint connected $1,2$-components. %

\begin{lemma}[Alon \cite{Alon91}]
  The probability that every connected $1,2$-component of biased hyperedges
  $T$ 
  in $G(H)$ has size at most $d \log(2 n)$ is at least $1/2$.
\end{lemma}

We call the first phase of Alon\rq{}s algorithm {\em successful} if there is no connected $1,2$-component of biased hyperedges 
$T$ of size greater than $d \log(2 n)$. 
Thus, the first phase is successful with probability at least $1/2$.
 Suppose that indeed the first phase was successful. Given the coloring from the first phase, the second phase
of the algorithm will alter this coloring to finally compute a bi-chromatic coloring, by recoloring vertices of $V$ that belong to biased hyperedges. 
We call a vertex from $V$ {\em bad} if it belongs to some biased hyperedge;
otherwise we call it {\em good}.
We call a hyperedge $N(w)$ {\em dangerous} if it contains at least $\alpha \delta$ bad vertices, i.e., vertices that belong to biased hyperedges. Note that biased hyperedges are also dangerous.
\\

\noindent
{\bf Observation:} If a hyperedge is not dangerous then it will not become monochromatic after the recoloring. \\

This follows because a non-dangerous hyperedge has at least  $\alpha \delta$ vertices of each color before the recoloring (because it is not biased), and during the recoloring less than $\alpha \delta$ of its vertices could change the color.
It follows, therefore, that during the recoloring phase we only need to worry about the dangerous hyperedges. The recoloring algorithm recolors all the vertices in biased hyperedges randomly and independently with equal probabilities. This means that any given dangerous hyperedge becomes monochromatic during this recoloring with probability less than $2^{-\alpha \delta}$. If we now look at the hypergraph $H$ only taking into account the vertices in the biased hyperedges, we can treat each such hyperedge as having size at least $\alpha \delta$ and therefore, by the assumption $2 e (d+1) \leq 2^{\alpha \delta}$, the Lov\'{a}sz Local Lemma implies that there exists a recoloring of these vertices such that the global coloring of the hypergraph $H$ is bi-chromatic.

Alon\rq{}s algorithm finds this recoloring of the vertices in the biased hyperedges by trying all such recolorings exhaustively, assuming that $d$ and $\delta$ are fixed constants. Namely, he shows that it suffices to recolor each maximal connected $1,2$-component of biased hyperedges exhaustively, independent from other such maximal connected $1,2$-components of biased hyperedges. By the above lemma, the size of each such connected $1,2$-component is $O(\log(2 n))$ if $d$ is constant, and, therefore there are only polynomially many recolorings for any such connected $1,2$-component.

The crucial point of Alon\rq{}s argument is that the vertices in the hyperedges of any maximal connected $1,2$-component $T$ of biased hyperedges can be recolored separately and independently from other such maximal connected $1,2$-components of biased hyperedges. This follows from the fact that there is no dangerous hyperedge that intersects hyperedges of two distinct such maximal connected $1,2$-components. 
Hence it suffices to recolor the vertices in hyperedges of each such connected $1,2$-component $T$ in such a way that no dangerous hyperedge that intersects a hyperedge of $T$ becomes monochromatic.

\subsubsection{Selected tools from Alon~\cite{Alon91}}

  Alon \cite{Alon91} proved the following result in general hypergraphs, without the assumption of $\delta, d$ being constants.
  
 \begin{theorem}[Alon \cite{Alon91}] \label{t:Alon-derand} Let $\delta \geq 2$, $d \geq 1$ be given integers and $0 < \alpha < 1$ with $d < 2^{ \delta/500}$. Then there exists a randomized algorithm that finds a bi-chromatic coloring of any given $\delta$-uniform hypergraph $H$ with $N$ hyperedges in which no hyperedge intersects more than $d$ other hyperedges, with expected running time $N^{O(1)}$. This algorithm can be derandomized implying a deterministic algorithm with running time $N^{O(1)}$.
 \end{theorem}
 
  We will now introduce some notions and tools from Alon\rq{}s paper \cite{Alon91} that we will use in our work. We call a set $T$ of hyperedges in $G(H)$ a {\em connected $2,3$-component} if $T$ is the set of hyperedges such that for any two hyperedges $x, y \in T$, $x \cap y = \emptyset$, and there is a path between them in the intersection graph $G(H)$ such that among any three consecutive hyperedges on this path at least one belongs to $T$.
  
\begin{lemma} \label{l:123-sets} Let $H$ be a $\delta$-uniform hypergraph with $N$ hyperedges in which no hyperedge intersects more than $d$ other hyperedges. We have the following:
\begin{itemize}
    \item    The probability that every connected $2,3$-component of biased hyperedges $T$ in $H$ has size at most $O(\log(N)/\delta)$ is at least $1/2$, where the probability is chosen with respect to a random coloring of $H$.
   \item  The probability that every connected $1,2$-component of biased hyperedges $T'$ in $H$ has size at most $O(d \log(N)/\delta)$ is at least $1/2$, where the probability is chosen with respect to $H$'s random coloring.
   \end{itemize}
\end{lemma}

\begin{proof}
  We execute the First Pass of Alon\rq{}s algorithm with $\alpha = 1/8$. That is we color each vertex of $H$ randomly and independently red or blue with equal probabilities. A hyperedge is biased if it has at most $\delta/8$ red vertices or at most $\delta/8$ blue vertices, and it is dengerous if it has at least $\delta/9$ vertices that belong to biased hyperedges. The algorithm will later recolor only vertices in biased hyperedges and hence every non-dangerous hyperedge will have at least $\delta/8 - \delta/9 = \delta/72$ vertices of each color.
  
   We will first bound the number of connected $2,3$-components of size $u$ in $G(H)$. Let us define a {\em special graph} on the set of vertices of graph $G(H)$ in which two vertices are adjacent if their distance in $G(H)$ is either $2$ or $3$. Every connected $2,3$-component on a set $T$ of vertices of $G(H)$ must contain a tree on $T$ in the special graph, and this graph has maximum degree at most $D = d^3$. From a well known fact, see \cite{Alon91}, that an infinite $D$-regular rooted tree contains $\frac{1}{(D-1)u + 1} {Du \choose u}$ rooted subtrees of size $u$, it follows that the number of trees of size $u$ containing one fixed vertex in a graph with maximum degree $D$ is less than this number, which is at most $(eD)^u$.
   
   Let $\cB$ be the set of all biased hyperedges in $H$. For any given connected $2,3$-component $T$ of size $u$, because its hyperedges form an independent set in $G(H)$, we have that $\Pr{T \subseteq \cB} \leq p^u$, where $p = 2 \cdot 2^{(h(\alpha)-1)\delta}$. Let $X$ be a random variable that denotes the number of connected $2,3$-components $T$ of biased hyperedges of size $u$. Using the assumption that $d < 2^{ \delta/500}$, we then obtain that
$$
   \Ex{X} \leq N(eDp)^u \leq N\left(2 e d^3 2^{(h(\alpha)-1)\delta}\right)^u \leq N\left(2^{3\delta/500 + 3 - (1-h(\alpha))\delta}\right)^u \leq N \cdot 2^{-c \delta u},
$$ where $c \approx 0.45$ and the last inequality holds for large enough $\delta > 0$. From this inequality we have that $\Ex{X} \leq 1/2$ for $u = \log(2N)/(c\delta)$. By applying the Markov inequality to $X$ we conclude that with probability at least $1/2$ there is no connected $2,3$-component of biased hyperedges of size greater than $u = \log(2N)/(c\delta)$.

Let us now consider any connected $1,2$-component $T'$ of biased hyperedges in $G(H)$. A maximal connected $2,3$-component $T$ of biased hyperedges within $T'$ must have a property that each hyperedge of $f_i \in T'$ is a neighbor in $G(H)$ of some hyperedge $f_j$ of $T$ (otherwise, this would contradict the maximality of $T$). Because the number of such neighbors $f_i$ of any given $f_j$ is at most $d$, we obtain that $|T'|/d \leq |T| \leq \log(2N)/(c\delta)$, and therefore
$$
  |T'| \leq d \log(2N)/(c\delta),
$$ and this concludes the proof of the lemma. 
\end{proof}

\noindent
{\bf Limited randomness.} It will be useful for us later to observe that to obtain the probabilistic conclusions of Lemma \ref{l:123-sets} it only suffices to use some restricted randomness. Namely, we only used $\delta$-wise independence of any subset of $\delta$ vertices among $n$ vertices that were colored randomly, to obtain that the probability that a fixed hyperedge (of size $\delta$) of $H$ is biased is at most $p = 2 \cdot 2^{(h(\alpha)-1)\delta}$. Then to show that with probability at least $1/2$ there are no connected $2,3$-components of biased hyperedges of size greater than $u = \log(2N)/(c\delta)$, we upper-bound $\Ex{X}$ as the sum over all fixed connected $2,3$-components of biased hyperedges of size $u$. Every such set of hyperedges has at most $u \delta = \log(2N)/c$ vertices, and thus $\log(2N)/c$-wise independence suffices in this case. 

This means that randomness is only needed for the first bullet point in Lemma \ref{l:123-sets}, and the conclusion in the second bullet point holds deterministically. We will later use this observation when dealing with connected $1,2$-components by using randomness only with respect to their maximal connected $2,3$-(sub)components. \\

After the First Pass, Alon considers each maximal connected $1,2$-component of biased hyperedges together with all dangerous hyperedges intersecting it, and call such part of the hypergraph a {\em piece}. Each piece is treated separately and independently and the algorithm recolors its vertices that belong to its biased hyperedges and at least $\delta/9$ vertices (that come from a biased hyperedge) in every other dangerous hyperedge of the piece. One can now repeat the previous argument by running the Second Pass to recolor
each piece separately. This idea partly inspires one of the steps of our approach, however, it is not sufficient in the LOCAL model (Alon showed that that technique was enough in the PRAM model). Namely, it does not have small enough failure probability for our derandomization purpose, and therefore we need to have not just two but many phases, and we need to carefully categorize various types of hyperedges that occur in our new process when the algorithm moves from phase to phase. 



\section{Algorithm \localalgo}

Consider the following algorithm.
As an input, each vertex in $V$ has a sequence of $x$ red-blue colors
from some red-blue colorings $C_1,\ldots,C_x$ of vertices in set $V$.
Later we will show in the analysis that there exists such a sequence of colorings guaranteeing successful bi-chromatic coloring of any admissible input hypergraph by the algorithm.
Let integer $u > 2$ be an input parameter of the algorithm.
We will show in the progress analysis in 
Section~\ref{sec:progress} that $u=O(1)$ suffices
to color bi-chromatically all hyperedges in any admissible 
hypergraph in $x=O(\log^2 n)$ communication rounds.

The algorithm proceeds in $x$ phases, split into $\beta=\Theta(\log n)$
consecutive epochs of $x/\beta$ phases each. 
Suppose that vertex $v\in V$ has not decided
its final color by the beginning of phase $i$, where $1\le i\le x$;
we assume that in the beginning of phase $1$ no vertex has decided
its color yet. We will use terms ``decided'' and ``fixed'' interchangeably, with respect to colors of vertices.
In phase $i$ vertex $v$ adopts temporarily (i.e., for the purpose of the current phase) its color from coloring $C_i$ and participates in the $y$-gossip protocol, together with other undecided vertices in $V$ and 
the hyperedges in $\cE$ to which they belong (recall that, according to the LOCAL model, a single round of communication is done between participating vertices and hyperedges they belong to).
When classifying hyperedges (biased or not, dangerous or not)
in this phase we take into account all its vertices:
those with colors fixed in previous phases and the remaining
with temporal colors taken from coloring $C_i$.

The $y$-gossip protocol, with $y = 6u+2$, runs as follows. In the beginning, each participating vertex assumes no knowledge about other participants and the topology between them, as well as their colors (the latter could be waived if we assume that vertices know whole colorings $C_1,\ldots,C_x$ from the beginning, not only their own vector of colors, but it depends on their local computational power - to be discussed later). In $y$ subsequent communication rounds, each participant (i.e., participating vertex or a hyperedge containing a participating vertex) broadcasts to all its neighbors its current knowledge about the topology of participants, and updates the topology based on the information received in this round from its participating neighbors.
At the end of the phase - if vertex $v$ finds itself good, i.e., outside of any biased hyperedge in $\cE$,
then it fixes its $i$-th color from the sequence as its final color (and stops participating in the next phases).
Otherwise, if $v$ finds itself in a maximal biased connected $1,2$-component of diameter smaller than $3u$ in the intersection graph, then it uses a Black-Box sequential algorithm to compute locally a
bi-chromatic recoloring of the whole biased connected $1,2$-component it belongs to, together with adjacent dangerous non-biased hyperedges. 

The motivation for this part is as follows. We will prove later in the progress analysis in Section~\ref{sec:progress} that it suffices for our algorithm to consider connected $1,2$-components $T$ whose maximal connected $2,3$-component, $R$, has size at most $u-1$, where $u$ is a fixed constant. This means that the diameter of $R$ is at most $3u-3$, and thus the diameter of $T$ is at most $3u-1$. Then if we take into account $T$ together with the intersecting dangerous non-biased hyperedges, denoted by $T'$, then the diameter of $T'$ is at most $3u+1$. Observe that the communication between two intersecting hyperedges with at least one participant in the intersection takes two rounds. This implies that $y = 6u+2$ communication rounds  suffices to discover such structure $T'$, and thus also to recognize that $T$ is the maximal connected $1,2$-component. Therefore the verification of existence of such $T$ of diameter smaller than $3u$ can be checked with $y = 6u + 2$ communication rounds.

The Black-Box algorithm re-colors only vertices without fixed color.
Note however that each hyperedge in the component, including the adjacent dangerous non-biased hyperedges, has at least $\alpha\delta$ such vertices.
Therefore, as we will justify later in Section~\ref{sec:Black-Box}, such a deterministic sequential algorithm exists and could be executed independently and in local memory by each participating vertex in the component or in the intersecting dangerous non-biased hyperedges. 
The common knowledge about this small component and the same Black-Box
algorithm simulated locally assure that these vertices compute the same
colorings and assign themselves colors consistently with this coloring.
Then vertex $v$ fixes the output color
as its final one and stops participating in the next phases.
On the top of that, if all hyperedges to which a vertex 
belongs to, have bi-chromatic fixed-color vertices, 
then it fixes its color and stops participating.
We call vertices that stop participating {\em idle vertices}.
Otherwise, it proceeds to the next phase.

A hyperedge in $\cE$ participates in the next phase only if 
its fixed-color part is monochromatic and
some of its vertices are still participating, i.e., have not fixed their colors in the previous phases, and they only participate in the $y$-gossip (by passing by the information about the topology and coloring
of participating vertices and hyperedges). Otherwise the hyperedge stops
participating and we call it {\em idle} from that point.

Note that once a vertex or a hyperedge becomes idle, it stays idle till
the end of the execution (because the colors of the vertex/vertices has/have been fixed).
Let $V_i$ denote the set of non-idle vertices in $V$ at the end of phase $i$,
and correspondingly, let $\cE_i$ be the set of non-idle hyperedges in $\cE$. 
For technical reason, we define $V_0=V$ and $\cE_0=\cE$.


\subsection{Black-Box coloring algorithm}
\label{sec:Black-Box}

Let us focus on a single maximal connected $1,2$-component $T$ of biased hyperedges together with its intersecting dangerous non-biased hyperedges in phase $i$; call it $T'$. Suppose also that all undecided vertices in $T'$ have found that $T$ is of diameter smaller than $3u$ in the intersection graph; 
later we will show in the progress analysis in Section~\ref{sec:progress} that indeed it suffices to consider connected $1,2$-components of diameter at most $3u - 1$ in the intersection graph, where $u$ is a fixed constant, and that $u-1$ will be the upper bound on the size of the maximal connected $2,3$-component of biased hyperedges in $T$. We note here that each undecided vertex $v$ explores a slightly larger diameter of $3u + 1$ around it (thus it uses $y = 6u+2$ communication rounds) to also discover $T'$ and to make sure that $T$ is indeed maximal --- it is the case if $v$ discovers at distance $3u + 1$ in the intersection graph only non-biased hyperedges.

We first observe that it follows from the Lov\'{a}sz Local Lemma, c.f., \cite{Alon91}, that there exists a bi-chromatic coloring of $T$ (coloring only undecided vertices), because $n \leq 2^{\alpha \delta}$. Then, the locally executed (sequential) coloring algorithm could be defined as follows. 
Suppose that $v$ is any undecided vertex 
in one of the hyperedges of $T'$. 
Based on the gathered knowledge about $T'$ during $y$ communication rounds, vertex $v$ uses any deterministic (polynomial time) hypergraph bi-chromatic coloring algorithm (for instance from \cite{Alon91,Heupler2013}) to color its local copy of $T'$. 
As mentioned before, only the undecided (i.e., participating) vertices are (re-)colored.
Then $v$ decides about its color to be the one assigned to it by this algorithm. This process is guaranteed to produce a correct bi-chromatic coloring of $T'$ by the facts that the coloring algorithm is deterministic, it operates on identical copies of $T'$, and that $T$ is the maximal connected $1,2$-component of biased hyperedges (that is, each hyperedge could be involved in at most one execution of Black-Box algorithm in a given phase).

\section{Analysis of algorithm \localalgo}

The analysis consists of two parts.
First we show structural properties of an execution of the algorithm,
including the mechanisms of changing status of hyperedges and vertices.
In the second part we prove the progress of the algorithm,
leading to successful completion of the coloring task within
$x$ phases.

\subsection{Structural properties of algorithm \localalgo}

In this part we analyze how the algorithm behaves when going from
one phase to the other. We start from analyzing what happens with hyperedges and vertices in the first phase,
based on the findings we formulate a useful invariant 
about partial coloring of hyperedges and prove it by 
induction.

Let us consider the first phase of our algorithm and let $T$ be one of the maximal connected $1,2$-components of biased hyperedges. %
%
%
If $T$ has diameter at most $3u-1$ in the original hypergraph
(i.e., the maximum distance between any two hyperedges of $T$
in the intersection graph is at most $3u-1$), then all hyperedges 
in $T$ are recolored by the Black-Box algorithm (applied in each
vertex covered by $T$); we call such component of {\em small-diameter}.
Otherwise, which we call a component of {\em large-diameter},
no vertex covered by $T$ triggers Black-Box and 
thus the hyperedges in $T$ stay non-colored until the next phase.

On the opposite side, if a hyperedge does not belong to any such
component $T$ nor intersects any, it is non-biased and non-dangerous,
all its vertices are good, thus they fixed their colors and become idle at the end of phase $1$,
together with the hyperedge (which becomes bi-chromatic and idle).

Suppose now that $f$ is a hyperedge which intersects one (or more) of hyperedges of $T$, but $f$ itself does not belong to $T$. We consider the following cases:
\begin{itemize}
\item 
$f$ is dangerous and not biased: since $f$ is dangerous, at least $\alpha \delta$ of its vertices are bad (i.e., belong to biased hyperedges). These hyperedges belong to the same maximal connected $1,2$-component of biased hyperedges, by definition of maximal connected $1,2$-component, and by our assumption this component must be $T$. We have now two sub-cases in our algorithm, depending on the diameter of this maximal connected $1,2$-component of biased hyperedges:
  \begin{itemize}
  \item 
  The biased hyperedge(s) that ``provide'' the bad vertices to $f$ belong to a small-diameter maximal connected $1,2$-component of biased hyperedges. These bad vertices will then be recolored, and as explained above, this recoloring is done (by Black-Box) in such a way that makes $f$ bi-chromatic. Even stronger, this recoloring makes bi-chromatic the intersection of the vertex set of $f$ with the set of all vertices that belong to biased hyperedges (that are being recolored), because it is of sufficiently large size of at least $\alpha\delta$. After this phase, our algorithm fixes the colors of those recolored vertices, and therefore from this phase onwards, $f$ will remain bi-chromatic and idle. 
  \item 
  The biased hyperedge(s) that ``provide'' the bad vertices to $f$ belong to a large-diameter maximal connected $1,2$-component of biased hyperedges. In this case in the current phase these bad vertices are not recolored. 
  \end{itemize} 
\item 
$f$ is dangerous and biased: impossible, because otherwise $f$ would belong to the maximal connected $1,2$-component $T$.
\item 
$f$ is not dangerous: this means that $f$ is not biased. 
This implies that there are at least $\alpha \delta$ red vertices and at least  $\alpha \delta$ blue vertices in $f$.
The fact that $f$ is not dangerous implies that 
it 
has less than $\alpha \delta$ bad vertices. 
We can now simply fix the colors of all good vertices in 
$f$.
%
Because our algorithm could only recolor (in the future) vertices that are bad in the current phase,
$f$ will remain bi-chromatic until the end of the algorithm.
\end{itemize}

Let us now consider any other phase $i>1$ of our algorithm. 
Recall that we denoted by $V_i\subseteq V$ the set of vertices with colors not fixed at the end of phase $i$, and by analogy, by $\cE_i\subseteq \cE$ the set of hyperedges that are monochromatically partly colored (i.e., non-idle) at the end of phase $i$. 

Assume that the following {\bf invariant} holds at the beginning of phase $i$: for each hyperedge $f\in \cE_i$
\begin{itemize}
	\item[(a)] 
$f\setminus V_i$ is monochromatically colored (fixed same color by all nodes in $f\setminus V_i$), and
	\item[(b)] 
	$|f\cap V_i|\ge \alpha\delta$.
\end{itemize}

We will show that the invariant also holds for phase $i+1\le x$.

The proof is a modification of the above analysis of phase one, by taking into account partial coloring (fixed colors by some nodes) restricted by the invariant.
We analyze several complementary cases and argue
that in each of them hyperedges will satisfy
the invariant at the end of phase $i+1$.

Let $T$ be one of the maximal connected $1,2$-component of biased hyperedges in phase $i+1$. Recall that each non-idle hyperedge 
in phase $i+1$ may consist of nodes with (same) color fixed in the previous phases and the remaining ones temporarily adopting colors from coloring $C_{i+1}$.

If $T$ has small-diameter (i.e., at most $3u-1$ in the original hypergraph),
then all hyperedges 
in $T$ are recolored by the Black-Box algorithm (applied in each vertex covered by $T$), as by the invariant for $i$,
each involved hyperedge has at least $\alpha\delta$
vertices that could be recolored (which is enough
to obtain bi-chromatic coloring, see Section~\ref{sec:Black-Box}). 
Thus, these hyperedges become idle at the end of phase $i+1$
and the invariant holds trivially for them.
Otherwise, 
no vertex covered by $T$ triggers Black-Box and 
thus the hyperedges in $T$ and their vertices stay exactly as they were in the beginning of phase $i+1$; 
thus, the hyperedges satisfy the invariant at the end of phase $i+1$
as they satisfied it at the end of phase $i$.

On the opposite side, if a hyperedge does not belong to any such
component $T$ nor intersects any, it is non-biased and non-dangerous,
all its vertices are good, thus those of them which are non-idle fix their colors and become idle at the end of phase $i+1$,
together with the hyperedge (which becomes bi-chromatic and idle).

Suppose now that $f$ is a hyperedge which intersects one (or more) of hyperedges of $T$, but $f$ itself does not belong to $T$. We consider the following cases:
\begin{itemize}
\item 
$f$ is dangerous and not biased: since $f$ is dangerous, at least $\alpha \delta$ of its vertices are bad (i.e., belong to biased hyperedges). These hyperedges belong to the same maximal connected $1,2$-component of biased hyperedges, by definition of maximal connected $1,2$-component, and by our assumption this component must be $T$. We have now two sub-cases in our algorithm, depending on the diameter of this maximal connected $1,2$-component of biased hyperedges:
  \begin{itemize}
  \item 
  The biased hyperedge(s) that ``provide'' the bad vertices to $f$ belong to a small-diameter maximal connected $1,2$-component of biased hyperedges. 
  In this case, consider vertices of $f$ that have not fixed their color. They will then be recolored by the Black-Box, together with the other similar vertices in the $1,2$-component and adjacent non-biased dangerous hyperedges, as the component is of small-diameter. 
  There are at least $\alpha\delta$ of them in $f$, by the invariant, therefore the recoloring makes them (and thus $f$) bi-chromatic, see Section~\ref{sec:Black-Box}.
  After this phase, our algorithm fixes the colors of those recolored vertices, and therefore from this phase onwards, $f$ will remain bi-chromatic and idle. 
  \item 
  The biased hyperedge(s) that ``provide'' the bad vertices to $f$ belong to a large-diameter maximal connected $1,2$-component of biased hyperedges. In this case in the current phase these undecided bad vertices are not recolored. 
  Thus, $f$ and its vertices stay exactly as they were in the beginning of phase $i+1$, and because they 
satisfied the invariant at the end of phase $i$,
they also satisfy it at the end of phase $i+1$.
%
  \end{itemize} 
\item 
$f$ is dangerous and biased: impossible, because otherwise $f$ would belong to the maximal connected $1,2$-component $T$.
\item 
$f$ is not dangerous: this means that $f$ is not biased. 
This implies that there are at least $\alpha \delta$ red vertices and at least  $\alpha \delta$ blue vertices in $f$.
The fact that $f$ is not dangerous implies that 
it 
has less than $\alpha \delta$ bad vertices. 
We can now simply fix the colors of all good vertices in 
$f$ which has not fixed their color before.
%
Because our algorithm could only recolor (in the future) vertices that are bad (and have not fixed their color yet) in the current phase,
$f$ will remain bi-chromatic until the end of the algorithm.
\end{itemize}

Thus, applying an inductive argument, we obtain the following fact.

\begin{lemma}
\label{lem:invariant}
The invariant holds at the end of any phase $i\le x$,
for any execution of algorithm \localalgo.
\end{lemma}


\subsection{Quantifying progress of algorithm \localalgo} 
\label{sec:progress}

Consider an admissible hypergraph $H$. We will recall some definitions here. A connected $2,3$-component is a set of hyperedges that is connected in the corresponding $2,3$-intersection graph. A $1,2$-dominating set of a given connected $1,2$-component is a subset of this component such that each hyperedge in the component is either in this subset or adjacent to it in the corresponding $1,2$-intersection graph. An $a,a+1$-intersection graph of hypergraph $H$, for an integer $a \geq 1$, is a graph whose vertices are hyperedges of $H$ and it has an edge between two of its vertices if their distance in the intersection graph of $H$ is at least $a$ and at most $a+1$.

A {\em $2,3$-skeleton of a given connected $1,2$-component} is a subset of hyperedges of the component forming a maximal connected $2,3$-component and being a $1,2$-dominating set with respect to the $1,2$-component. For an arbitrary (not necessarily connected) $1,2$-component, a $2,3$-skeleton is a collection of $2,3$-skeletons over all connected parts of the component.

We will use the notion of a $2,3$-skeleton mainly in context of a $1,2$-component of biased hyperedges, that is, with respect to a given coloring (phase) and where all hyperedges of such $2,3$-skeleton are biased. However, sometimes, we will also use the notion of a $2,3$-skeleton in the context of only a given topology of a hypergraph and without assuming existence of a coloring.


Consider all admissible hypergraphs of $n$ vertices and $n$ hyperedges, and partition them into classes $\cF_k$, for any $k\le n$,
such that $\cF_k$ contains all admissible hypergraphs with the maximum size $2,3$-skeleton of size equal to $k$. Let $\cS_k$ be the set of all such possible $2,3$-skeletons of size $k$ in all admissible topologies of hypergraphs in $\cF_k$. For each $S\in \cS_k$, let $J_S$ be the set of all hypergraphs 
in $\cF_k$ having $S$ among their maximum size $2,3$-skeletons.
Note that $|\cS_k|\le \binom{n}{k}$.

Let us partition the execution of the algorithm \localalgo~into consecutive epochs, each consisting of $\beta$ subsequent phases, where we will use $\beta = \log n$. We call a hyperedge {\em successful in the epoch} if by the end of this epoch it had status non-biased or biased and belonging to a small (i.e., of size smaller than $u$ -- our fixed constant) $2,3$-skeleton of some 
maximal connected $1,2$-component of biased hyperedges; otherwise we call it {\em unsuccessful}.
   We will first need the following technical claim.

\begin{claim}\label{c:prob-cond}
 For any fixed hyperedge in the input hypergraph, the probability that it is unsuccessful in the consecutive phases $1,2,\ldots, \ell$ with a given epoch is at most $2^{-c' \alpha \delta \ell}$, where $\ell \leq \beta$ and $c' > 0$ is an absolute constant independent from $n$.
\end{claim}

\begin{proof}
  We will show here an argument for $\ell = 2$, which then can easily be extended to any $\ell$. Let $A_i$ denote the event that the fixed hyperedge $f$ is unsuccessful in phase $i$, meaning that $f$ is one of the biased hyperedges belonging to a large (i.e., of size at least $u$) $2,3$-skeleton of some connected $1,2$-component of biased hyperedges with respect to the coloring of phase $i$.
  
  As we argue in the proof of Lemma \ref{l:123-sets}, $f$ is unsuccessful in the first phase with probability at most
  $$
    \Pr{A_1} \leq n \cdot 2^{-c \delta u}
    \ ,
  $$ 
  for some constant $0 < c < 1$. Thus, this probability is also at most $2^{-c' \delta u}$, where we assume that the constant $c'$ is only slightly smaller than $c$. And to be even more conservative, this probability is also at most $2^{-c' \delta}$ because $u>2$ (this uses only randomness of colorings of the vertices in $f$). 
  
  Now, consider the second phase and the event $A_2$. This event can only occur if event $A_1$ occurred in the previous phase. We need to show an upper bound on the probability of the event $A_1 \cap A_2$, and we do it using the following observation:
  $$
    \Pr{A_1 \cap A_2} = \Pr{A_1} \cdot \frac{\Pr{A_1 \cap A_2}}{\Pr{A_1}} = \Pr{A_1} \cdot \Pr{A_2 | A_1}
    \ .
  $$ 
  By our algorithm, 
  only a subset of hyperedge $f$ (and of other hyperedges) of size (at least) $\alpha \delta$ is (are) ``transferred" to the second phase in order to be colored by new coloring, c.f., Lemma~\ref{lem:invariant}. Therefore, we can (conservatively) upper bound the conditional probability $\Pr{A_2 | A_1}$ by using the same argument as for $\Pr{A_1}$ as follows:
  $$
    \Pr{A_2 | A_1} \leq 2^{-c' \alpha \delta}
    \ ,
  $$ 
  which implies that 
  $$
    \Pr{A_1 \cap A_2} \leq 2^{-2 c' \alpha \delta}
    \ .
  $$ 
  Repeating this reasoning $\ell-1$ times we will obtain that
  $$
    \Pr{A_1 \cap A_2 \cap \cdots \cap A_{\ell}} \leq 2^{- \ell c' \alpha \delta}
    \ .
  $$
\end{proof}

\begin{lemma}
\label{lem:localalgo}
The probability of the event within an epoch that there exists an input hypergraph in $\cF_k$ with a set 
of more than $k/2$ unsuccessful hyperedges forming a $2,3$-skeleton of some $1,2$-component of biased hyperedges in the input hypergraph, is at most $2^{-c''k\alpha\delta\beta}$, for some absolute constant $c''>0$ which is independent from $n$.
\end{lemma}

\begin{proof}
The proof is for the first epoch only, but could easily be generalized into next epochs.

Consider a set $S\in \cS_k$.
The probability that more than half of hyperedges in $S$ will be unsuccessful in an epoch is at most
\[
\binom{k}{k/2} \rho^{k/2}
\ ,
\]
where $\rho = 2^{-c'\alpha\delta \beta}$ by Claim~\ref{c:prob-cond}.
This could further be upper bounded by
\[
2^{k-(\log k)/2-(c'/2)k\alpha\delta\beta}
\ .
\]

Next, the probability of the event, call it $L_0$,
that there is $S\in \cS_k$
such that more than half of hyperedges in $S$
will be unsuccessful in the epoch, is at most
\[
|\cS_k| \cdot 2^{k-(\log k)/2-(c'/2)k\alpha\delta\beta}
\le
2^{k\log(en/k) + k-(\log k)/2-(c'/2)k\alpha\delta\beta}
\le
2^{-c''k\alpha\delta\beta}
\ ,
\]
since $|\cS_k|\le \binom{n}{k}$ and $\log(en/k) + 1-(c'/2)\alpha\delta\beta<-c''\alpha\delta\beta$ due to definitions of $\alpha,\beta,\delta$ and for some suitable
constant $0<c''<c'/2$.

Finally, we prove that the event of the lemma, let us denote this event by $L_1$, holds with probability at most $2^{-c''k\alpha\delta\beta}$.
We prove it by showing that this event is contained in event
$L_0$.
Consider a hypergraph in $\cF_k$. It belongs to some sets $J_S \subseteq \cF_k$ with $S\in \cS_k$, namely, to those where $S$ is a maximum
size $2,3$-skeleton for the hypergraph. 
If the event $L_1$ holds, there is a $2,3$-skeleton $S$ of $k$ hyperedges such that the hypergraph is in $J_S$
and more than $k/2$ hyperedges in $S$ are biased.
This means that the event $L_0$ holds.
Therefore, the probability of $L_1$ is upper bounded by the
computed above probability of $L_0$, which is 
$2^{-c''k\alpha\delta\beta}$.

Note that here hyperedges that are non-biased and dangerous are not part of such a $2,3$-skeleton, as they lay in a boundary of $1,2$-components, and therefore a maximum size $2,3$-skeleton contains only edges that are biased.
\end{proof}

\begin{theorem}
\label{thm:localalgo}
For any parameter $n$ and any admissible hypergraph of $n$ vertices and $n$ hyperedges, Algorithm \localalgo{} instantiated with local random colorings terminates successfully in $\beta\log n$ rounds with positive probability, where $\beta = \Theta(\log n)$.
\end{theorem}

\begin{proof}
Consider parameter $n$, which is given as an input to the algorithm.
All admissible hypergraphs of $n$ vertices and $n$ hyperedges can be partitioned into classes $\cF_k$, for $1\le k\le n$.
It follows from Lemma~\ref{lem:localalgo} that in a single epoch, algorithm \localalgo{}
colors partly all hypergraphs from $\cF_k$
in such a way that the remaining non-idle parts
belong to $\bigcup_{i\le k/2} \cF_i$,
with probability at least $1 - 2^{-c''k\alpha\delta\beta}$,
for some constant $c''>0$.
We call it {\em progressive partial coloring}.
By the union bound, algorithm \localalgo{} 
does such progressive partial coloring in a single epoch on any admissible hypergraph of $n$ vertices and $n$
hyperedges with probability at least
\[
1-\sum_{k=1}^n 2^{-c''k\alpha\delta\beta}
>
1-\sum_{k=1}^n 2^{-2k\log n}
>
1-1/(2\log n)
\ ,
\]
where the first inequality holds because of the definitions of $\alpha,\beta,\delta$ that imply $\alpha\delta\beta> (2\log n)/c''$.

Finally, taking the union bound of the above events
over $\beta=\log n$ epochs, we obtain that in all epochs such progressive coloring is done for every admissible hypergraph with probability greater than
$1-\beta \cdot 1/(2\log n) = 1/2$.
This means that for every admissible hypergraph, the $2,3$-skeleton of its remaining non-idle part shrinks by
a factor of $1/2$ in every epoch, and thus reaches $2$ before the last 
epoch and could be directly re-colored by the end of the last epoch (because the size of the largest $2,3$-skeleton is at most $2$, which is smaller than $u$; recall that parameter $u$ was assumed to be strictly greater than $2$). All this holds with probability greater than $1/2$, which completes the proof of the theorem.
\end{proof}

Using the probabilistic argument, the following could be derived
from Theorem~\ref{thm:localalgo}.

\begin{corollary}
\label{cor:final}
There is an input set of $\beta\log n$ colorings such that algorithm \localalgo{} instantiated with these colorings 
terminates successfully in $\beta\log n$ rounds.
\end{corollary}

Note that these colorings from Corollary~\ref{cor:final} could be computed locally and independently by each vertex in the beginning of the algorithm 
in zero cost (in the beginning of the first round), e.g., 
by simulating the algorithm in local memory for all possible colorings and hypergraph topologies, and choosing lexicographically the first colorings
that are good for all admissible hypergraphs.
As mentioned, even though the local computation could be complex,
it does not affect the polylogarithmic communication round complexity.

\section{Conclusions and open problems}

An open question emerging directly from this work is
how to compute colorings $C_i$ locally in an efficient way, e.g., in polynomial time in local memory.
Our algorithm uses $\Theta(\log^2 n)$ colorings,
and what is more, these colorings need to be verified
against each admissible hypergraph. How to do it substantially faster than a naive search is a challenging problem.

Another interesting open problem would be to consider a congest model, in which nodes can exchange only a limited knowledge in each round, or other communication models with restrictions such as limitation on the number of messages
(c.f.,~\cite{GGK11}) or interference (c.f.,~\cite{CKPR11}). In our algorithm, we need to 
collect information about connected biased components that
could be big in size, thus impossible to achieve with
restricted message size.

\end{document}